\documentclass[11pt]{article}

\usepackage[utf8x]{inputenc}
\usepackage{amsmath}
\usepackage{amssymb}
\usepackage{amsthm}
\usepackage{algorithmic}
\usepackage{graphicx}
\usepackage{algorithm}
\usepackage{verbatim}

\usepackage{marginnote}

\usepackage{a4wide}

\newcommand{\eps}{\varepsilon}

\newcommand{\LAS}{\text{\sc{Las}}}

\newcommand{\set}[1]{\left\{ #1 \right\}}
\newcommand{\PS}{\mathcal{P}}

\DeclareMathOperator*{\adjugate}{adj}

\newtheorem{theorem}{Theorem}
\newtheorem{lemma}[theorem]{Lemma}

\newtheorem{corollary}[theorem]{Corollary}
\newtheorem{definition}{Definition}

\newtheorem{remark}{Remark}

\begin{document}
\title{\textbf{On the Hardest Problem Formulations \\for the $0/1$ Lasserre Hierarchy}~\thanks{This is the full version of the paper that was presented at ICALP 2015.}}


\author{Adam Kurpisz \and Samuli Lepp\"anen \and Monaldo Mastrolilli\\
{\small \textit{Dalle Molle Institute for Artificial Intelligence Research (IDSIA),}}\\ 
{\small \textit{6928 Manno, Switzerland,}}\\ 
{\small \textit{\{adam,samuli,monaldo\}@idsia.ch}}
}
\date{}
\maketitle
\begin{abstract}
The Lasserre/Sum-of-Squares (SoS) hierarchy is a systematic procedure for constructing a sequence of increasingly tight semidefinite relaxations. It is known that the hierarchy converges to the 0/1 polytope in $n$ levels and captures the convex relaxations used in the best available approximation algorithms for a wide variety of optimization problems.

In this paper we characterize the set of 0/1 integer linear problems and unconstrained 0/1 polynomial optimization problems that can still have an integrality gap at level $n-1$. These problems are  the hardest for the Lasserre hierarchy in this sense.


\end{abstract}

\section{Introduction}
The \emph{Sum of Squares} (SoS) proof system introduced by Grigoriev and Vorobjov~\cite{GrigorievV01} is a proof system based on the \emph{Positivstellensatz}. Shor~\cite{schor87}, Nesterov~\cite{Nesterov00}, Parrilo~\cite{parrilo00} and Lasserre~\cite{Lasserre01} show that it can be efficiently automatized using semidefinite programming (SDP) such that any $n$-variable degree-$d$ proof can be found in time $n^{O(d)}$. The SDP, often called the Lasserre/SoS\footnote{For brevity, we will interchange Lasserre hierarchy with SoS hierarchy since they are essentially the same in our context.} hierarchy, is the dual of the SoS proof system, meaning that the Lasserre hierarchy value at ``level d/2'' of an optimization problem is equal to the best provable bound using a degree-$d$ SoS proof (see the monograph by Laurent~\cite{laurent09}). For a brief history of the different formulations from \cite{GrigorievV01}, \cite{Lasserre01}, \cite{parrilo00} and the relations between them and results in real algebraic geometry we refer the reader to \cite{ODonnellZ13}.

The Lasserre hierarchy can be seen as a systematic procedure to strengthen a relaxation of an optimization problem by constructing a sequence of increasingly tight SDP relaxations. The tightness of the relaxation is parametrized by its \emph{level} or \emph{round}, which corresponds to the degree of the proof in the proof system. 
Moreover,
it captures the convex relaxations used in the best available approximation algorithms for a wide variety of optimization problems. For example, the first round of the hierarchy for the \textsc{Independent Set} problem implies the Lov\'{a}sz $\theta$-function~\cite{Lovasz79} and for the \textsc{Max Cut} problem it gives the Goemans-Williamson relaxation~\cite{GoemansW95}. The ARV relaxation of the \textsc{Sparsest Cut} \cite{AroraRV09} problem is no stronger than the relaxation given in the third round of the Lasserre hierarchy, and the subexponential time algorithm for \textsc{Unique Games}~\cite{AroraBS10} is implied by a sublinear number of rounds~\cite{BarakRS11,GuruswamiS11}. More recently, it has been shown that $O(1)$ levels of the Lasserre hierarchy is equivalent in power to any polynomial size SDP extended formulation in approximating maximum constraint satisfaction problems \cite{LeeRagSteu15}.
Other approximation guarantees that arise from the first $O(1)$ levels of the Lasserre (or weaker) hierarchy can be found in~\cite{BarakRS11,BateniCG09,Chlamtac07,ChlamtacS08,DBLP:conf/soda/CyganGM13,VegaK07,GuruswamiS11,MagenM09,RaghavendraT12}. For a more detailed overview on the use of hierarchies in approximation algorithms, see the surveys~\cite{Chla12,Laurent03,laurent09}.



The limitations of the Lasserre hierarchy have also been studied.
Most of the known lower bounds for the hierarchy originated in the works of Grigoriev~\cite{Grigoriev01,Grigoriev01b} (also independently rediscovered later by Schoenebeck~\cite{Schoenebeck08}).
In \cite{Grigoriev01b} it is shown that random 3XOR or 3SAT instances cannot be solved by even $\Omega(n)$ rounds of SoS hierarchy. Lower bounds, such as those of~\cite{BhaskaraCVGZ12,Tulsiani09} rely on \cite{Grigoriev01b,Schoenebeck08} plus gadget reductions. For a different technique to obtain lower bounds, see the recent paper \cite{BarakCK15}.

A particular weakness of the hierarchy revolves around the fact that it has hard time reasoning about terms of the form $x_1 + ... + x_n$ using the fact that all $x_i$'s are 0/1.
Grigoriev~\cite{Grigoriev01} showed that $\lfloor n/2 \rfloor$ levels of Lasserre are needed to prove that the polytope $\{x\in [0,1]^n| \sum_{i=1}^n x_i= \lfloor n/2 \rfloor+1/2\}$ contains no integer point. A simplified proof can be found in~\cite{GrigorievHP02}.

 In~\cite{Cheung07} Cheung considered a simple instance of the \textsc{Min Knapsack} problem, i.e. the minimization of $\sum_{i=1}^n x_i$ for 0/1 variables such that $\sum_{i=1}^n x_i \geq \delta(n)$, for some $\delta(n) <1$ that depends on $n$. Cheung proved that the Lasserre hierarchy requires $n$ levels to converge to the integral polytope. This is shown by providing a feasible solution at level $n-1$ of value $\frac{n}{n+1}$, whereas the smallest integral solution has value $1$. This gives an integrality gap\footnote{The \emph{integrality gap} is defined to be the measure of the quality of the relaxation described by the ratio between the optimal integral value and the relaxed optimal value. If this ratio is different from 1 we will say that ``there is an integrality gap''.}
 of $1+\frac{1}{n}$ that vanishes with~$n$.

We emphasize that the main interest in the work of Cheung revolves around understanding how fast the Lasserre hierarchy converges to the integral polytope and not how fast the integrality gap reduces, therefore not ruling out the possibility that the integrality gap might decrease slowly with the number of levels.
This is conceptually an important difference.
For the $\textsc{Max Knapsack}$ (or $\textsc{Min Knapsack}$) problem the presence of an integrality gap at some ``large'' level $t(n)$, that depends on $n$, is promptly implied by $P\not= NP$, whereas the existence of a ``large'' integrality gap at some ``large'' level $t(n)$ is not immediately clear (since both $\textsc{Max Knapsack}$ and $\textsc{Min Knapsack}$ problems admit an FPTAS).
With this regard, note that Cheung's result also implies that for the \textsc{Max Knapsack} the Lasserre hierarchy requires $n$ levels to converge to the integral polytope. However, in~\cite{KarlinMN11} it is shown that only $O(1/\eps)$ levels are needed to obtain an integrality gap of $1-\eps$, for any arbitrarily small constant $\eps>0$. It is also worth pointing out that currently the Cheung knapsack result~\cite{Cheung07} is the only known integrality gap result for Lasserre/Sum-of-Squares hierarchy at level $n-1$.

%
%
\paragraph{Our results.}
With $n$ variables, the $n$-th level of the Lasserre hierarchy is sufficient to obtain the $0/1$ polytope, where the only feasible solutions are convex combinations of feasible integral solutions~\cite{Lasserre01}.
This can be proved by using the \emph{canonical lifting lemma} (see Laurent~\cite{Laurent03}), where the feasibility of a solution to the Lasserre relaxation at level $n$ reduces to showing that a certain diagonal matrix is positive semidefinite (PSD).

The main challenge in analyzing integrality gap instances at level smaller than $n$ is showing that a candidate solution satisfies the positive semidefinite constraints.
In this paper, we first show that the feasibility of a solution to the Lasserre relaxation at level $n-1$ reduces to showing that a matrix differing from a diagonal matrix by a rank one matrix (almost diagonal form) is PSD. We analyze the eigenvalues of the almost diagonal matrices and obtain compact
necessary and sufficient conditions for the existence of an integrality gap of the Lasserre relaxation at level $n-1$. This result can be seen as the opposite of~\cite{GouveiaPT13} where they consider the case when the first order Lasserre relaxation is exact.

Interestingly, for 0/1 integer linear programs the existence of a gap at level $n-1$ implies that the problem formulation contains only constraints of the form we call \emph{Single Vertex Cutting} (SVC). An SVC constraint only excludes one vertex of the $\{0,1\}^n$ hypercube. It can thus be seen as the most generic non-trivial form of constraint, since the feasible set of any integer linear program can be modeled using only constraints of this form. 

This characterization allows us to show that $n$ levels of Lasserre are needed to prove that a polytope defined by (exponentially many) SVC constraints contains no integer point.
No other example of this kind was known at level $n$ (the previously known example in~\cite{Grigoriev01} requires $\lfloor n/2 \rfloor$ levels).


One problem where SVC constraints can arise naturally is the \textsc{Knapsack} problem. By applying the computed conditions, we improve the Cheung~\cite{Cheung07} \textsc{Min Knapsack} integrality gap of the Lasserre relaxation at level $n-1$ from $1+1/n$ to any arbitrary large number. 
This shows a substantial difference between the \textsc{Min Knapsack} and the \textsc{Max Knapsack} when we take into consideration the integrality gap size of the Lasserre relaxation.

Furthermore, we show that a similar result holds beyond the class of integer linear programs. More precisely, we show that any unconstrained 0/1 polynomial optimization problem exhibiting an integrality gap at level $n-1$ of the Lasserre relaxation has necessarily an objective function given by a polynomial of degree~$n$. This rules out the existence of any integrality gap at level $n-1$ for any $k$-ary boolean constraint satisfaction problem with $k<n$. Finally, we provide an example of an unconstrained 0/1 polynomial optimization problem with an integrality gap at level $n-1$ of the Lasserre hierarchy, 
and discuss why the problem can be seen 
as a constraint satisfaction version of an SVC constraint. Our result complements the recent paper~\cite{FawziSaundersonParrilo15} where it is shown that the Lasserre relaxation 
does not have any gap 
at level $\lceil \frac{n}{2} \rceil$ when optimizing 
$n$-variate 
0/1 polynomials of degree 2.



\section{The Lasserre Hierarchy}
In this section we provide a definition of the Lasserre hierarchy~\cite{Lasserre01}.
For the applications that we have in mind, we restrict our discussion to optimization problems with $0/1$-variables and linear constraints.
More precisely, we consider the following general optimization problem $\mathbb{P}$: Given a multilinear polynomial $f:\{0,1\}^n\rightarrow \mathbb{R}$

\begin{equation}\label{eq:polyproblem}
\mathbb{P}: \quad \min\{f(x)| x\in\{0,1\}^n,  g_{\ell}(x)\geq 0 \text{ for } \ell\in [m]\}
\end{equation}
where $\{ g_{\ell}(x): \ell\in [m]\}$ are linear functions of $x$.

Many basic optimization problems are special cases of $\mathbb{P}$. For example, any $k$-ary boolean constraint satisfaction problem, such as \textsc{Max Cut}, is captured by~\eqref{eq:polyproblem} where a degree $k$ function $f(x)$ counts the number of satisfied constraints, and no linear constraints $ g_{\ell}(x)\geq 0$ are present. Also any $0/1$ integer linear program is a special case of~\eqref{eq:polyproblem}, where $f(x)$ is a linear function.


Lasserre~\cite{Lasserre01} proposed a hierarchy of SDP relaxations for increasing $\delta$,
\small{
\begin{equation}\label{eq:lass1}
 \min\{L(f)| L: \mathbb{R}[X]_{2\delta}\rightarrow \mathbb{R},  L(1)=1, \text{ and } L(u^2), L(u^2 g_{\ell})\geq 0,  \forall \text{ polynomial } u \}
\end{equation}
where $L: \mathbb{R}[X]_{2\delta}\rightarrow \mathbb{R}$ is a linear map with $\mathbb{R}[X]_{2\delta}$ denoting the ring $\mathbb{R}[X]$ restricted to polynomials of degree at most $2\delta$.\footnote{In \cite{BarakBHKSZ12}, $L(p)$ is written $\tilde{ \mathbb{E}}[p]$ and called the ``pseudo-expectation'' of $p$.} In particular for $0/1$ problems $L$ vanishes on the truncated ideal generated by $x_i^2-x_i$. Note that~\eqref{eq:lass1} is a relaxation since one can take $L$ to be the evaluation map $f\rightarrow f(x^*)$ for any optimal solution $x^*$.

Relaxation~\eqref{eq:lass1} can be equivalently formulated in terms of \emph{moment matrices}~\cite{Lasserre01}. In the context of this paper, this matrix point of view is more convenient to use and it is described below. In our notation we mainly follow the survey of Laurent~\cite{Laurent03} (see also \cite{Rot13}).

\paragraph{Variables and Moment Matrix.} Throughout this paper, vectors are written as columns. Let $N$ denote the set $\{1,\ldots,n\}$. The collection of all subsets of $N$ is denoted by $\PS(N)$. For any integer $t\geq 0$, let $\PS_t(N)$ denote the collection of subsets of $N$ having cardinality at most~$t$.
Let $y\in \mathbb{R}^{\PS(N)}$. For any nonnegative integer $t\leq n$, let $M_t(y)$ denote the matrix with $(I,J)$-entry $y_{I\cup J}$ for all $I,J\in \PS_t(N)$. Matrix $M_t(y)$ is termed in the following as the \emph{t-moment matrix} of $y$. For a linear function $g(x) = \sum_{i=1}^n g_{i} \cdot x_i + g_0$, we define $g*y$ as a vector, often called \emph{shift operator}, where the $I$-th entry is $(g*y)_I=\sum_{i=1}^n g_i y_{I\cup\{i\}} + g_0 y_I$. Let $f$ denote the vector of coefficients of polynomial $f(x)$ (where $f_I$ is the coefficient of monomial $\Pi_{i\in I}x_i$ in $f(x)$).


\begin{definition}\label{lassDef}
The Lasserre relaxation of problem~\eqref{eq:polyproblem} at the $t$-th level, denoted as $\LAS_t(\mathbb{P})$, is the following

\begin{equation}
\LAS_t(\mathbb{P}):\quad \min\left\{\sum_{I \subseteq N} f_I y_I | y\in \mathbb{R}^{\PS_{2t+2d}(N)} \text{ and } y\in \mathbb{M} \right\}
\end{equation}
where $\mathbb{M}$ is the set of vectors $y\in \mathbb{R}^{\PS_{2t+2d}(N)}$ that satisfy the following PSD conditions
\begin{eqnarray}
y_{\varnothing}&=&1  \\
M_{t+d}(y)&\succeq& 0  \\
M_{t}( g_{\ell}*y )&\succeq& 0 \qquad \ell\in [m]
\end{eqnarray}
where $d=0$ if $m=0$ (no linear constraints) otherwise $d=1$.
\end{definition}

We will use the following known facts (see e.g.~\cite{Laurent03,Rot13}).
Consider any vector $w\in\mathbb{R}^{\PS(N)}$ (vector $w$ is intended to be either the vector $y\in\mathbb{R}^{\PS(N)}$ of variables or the shifted vector $g*y$ for any $g\in\mathbb{R}^{\PS(N)}$).
For any $I\in \PS(N)$, variables $\{w_{I}^N:I\subseteq N\}$ are defined as follows:
\begin{equation*} \label{eq:probevent}
w_{I}^N: =\sum_{H\subseteq N\setminus I} (-1)^{|H|} w_{H\cup I}
\end{equation*}
Note that $w_I = \sum_{I\subseteq J} w_J^N$ (by using inclusion-exclusion principle, see~\cite{Rot13}). The latter with $y_{\emptyset}=1$ implies that $\sum_{J\subseteq N} y_J^N =1$, and that the objective function can be rewritten as follows:
$$\sum_{I \subseteq N} f_I y_I=\sum_{I\subseteq  N} f(x_I) y_I^N$$
where $f(x_I)$ denotes the value of $f(x)$ when $x_i=1$ for $i\in I$ and $x_i=0$ for $i\not\in I$.

Congruent transformations are known not to change the sign of the eigenvalues (see e.g. \cite{HornJohnson03}). It follows that in studying the positive-semidefiniteness of matrices we can focus on congruent matrices without loss of generality.
Let $D_{t}(w)$ denote the diagonal matrix in $\mathbb{R}^{\PS_{t}(N)\times \PS_{t}(N)}$ with $(I,I)$-entry equal to $w_I^N$ for all $I\in \PS_{t}(N)$.
\begin{lemma}\cite{Laurent03}\label{prob}
Matrix $M_{n}(w)$ is congruent
to the diagonal matrix
$
D_{n}(w)
$.
\end{lemma}
By Lemma~\ref{prob}, $M_n(y)\succeq 0$ implies that the variables in $\{y_I^N:I\subseteq N\}$ can be interpreted as a probability distribution (see~\cite{Laurent03,Rot13}), where $y_I^N$ is the probability that the variables with index in $I$ are set to one and the remaining to zero.
\begin{lemma}\cite{Laurent03}\label{th:constentry}
For any polynomial $g$ of degree at most one, $y\in \mathbb{R}^{\PS(N)}$ and $z = g*y$ we have
$
z^N_I = g(x_I)\cdot y^N_I
$
where $g(x_I) = \sum_{i\in I} a_i + b$. 
\end{lemma}

Note that, by using Lemma~\ref{prob} and Lemma~\ref{th:constentry}, it can be easily shown the well known fact that at level $n$ any solution can be written as a convex combination of feasible integral solutions. The latter implies that any integrality gap vanishes at level $n$.

\section{The $(n-1)$-Moment Matrix}
In the following we show that $M_{n-1}(w)$ is congruent to the diagonal matrix $D_{n-1}(w)$ perturbed by a rank one matrix, and analyze its eigenvalues.
For ease of notation, we will use $D$ to denote $D_{n-1}(w)$ throughout this section.

%
\begin{lemma}\label{almostdiag}
Matrix $M_{n-1}(w)$ is congruent to the matrix
$
D + w_N^N \cdot vv^{\top}
$, 
 where $v$ is a $|\PS_{n-1}(N)|$-dimensional vector with $v_I= (-1)^{n+1-|I|}$ for any $I\in \PS_{n-1}(N)$.
\end{lemma}
\begin{proof}
Let $Z_{n-1}$ denote the \emph{zeta matrix} of the lattice $\PS_{n-1}(N)$, that is the square $0$-$1$ matrix indexed by $\PS_{n-1}(N)$ such that $[Z_{n-1}]_{I,J}=1$ if and only if $I \subseteq J$.
%
%
This matrix is known to be invertible (note that it is upper triangular with unit diagonal entries) and the inverse is known as the M\"{o}bius matrix of $\mathcal{P}_{n-1}(N)$ whose entries are defined as follows:
\begin{equation}\label{mobiusmatrix}
\left[Z_{n-1}^{-1}\right ]_{I,J}=
\left\{
\begin{array}{ll}
(-1)^{|J\setminus I |}   & \text{if } I \subseteq J,\\
0 & \text{otherwise}.
\end{array} \right.
\end{equation}
Since $w_I = \sum_{I\subseteq J} w_J^N$ by direct inspection we have that
\begin{eqnarray*}
M_{n-1}(w)= Z_{n-1}D Z_{n-1}^{\top} + w_N^N \cdot \mathbb{J}
\end{eqnarray*}
where $\mathbb{J}$ is the all-ones matrix. By multiplying both sides by the M\"{o}bius matrix we obtain that
$M_{n-1}(w)$ is congruent to $D + w_N^N \cdot (Z_{n-1}^{-1}e) (Z_{n-1}^{-1} e)^{\top}$, where $e$ is the vector of all-ones, and the claim follows.

\end{proof}

\subsection{Positive semidefiniteness of $M_{n-1}(y)$}

In this section we derive the necessary and sufficient conditions for $M_{n-1}(w) \succeq 0$. From Lemma~\ref{almostdiag} we have that
$
M_{n-1}(y) \succeq 0 \Leftrightarrow D + w_N^N vv^{\top} \succeq 0
$, 
where $vv^{\top}$ is a rank one matrix with entries $\pm 1$. 
%
\begin{lemma} \label{lemma:PSD_condition}
If $w_N^N \neq 0$ then, for any $I\subseteq N$, $\lambda = w_I^N$ is an eigenvalue of the matrix $D + w_N^N vv^\top$ if and only if there is another $J\not = I$ with $w_I^N=w_J^N$, $J\subseteq N$; The remaining eigenvalues are the solutions $\lambda$ of the following equation
\begin{equation} \label{eq:essential_condition}
\sum_{N\neq I \subseteq N} \frac{1}{\lambda - w_I^N} = \frac{1}{w_N^N}
\end{equation}

\end{lemma}
\begin{proof}
Consider the zeroes $\lambda$ of the characteristic polynomial of $D + w_N^N vv^\top$:
\begin{equation} \label{eq:characteristic_polynomial}
\det(\lambda I - (D+w_N^N vv^\top)) = \det(D_\lambda - w_N^N vv^\top) = 0
\end{equation}
where $D_{\lambda}=\lambda I -D$. Applying Cauchy's formula for the determinant of a rank-one pertubation \cite[p.~26]{HornJohnson03} we can write this as
\begin{equation} \label{eq:adjoint_equality}
\det(D_\lambda) - w_N^N v^\top\adjugate(D_\lambda) v = 0
\end{equation}
Consider a solution $\lambda$ to \eqref{eq:adjoint_equality}. Exactly one of the following three cases must hold:
\begin{enumerate}
 \item $D_\lambda$ is nonsingular, meaning that $\lambda \neq w_I^N$ for all $N\neq I \subseteq N$. Then $\adjugate(D_\lambda) = (\det D_\lambda)D^{-1}_\lambda$ 
 and the above becomes
$$
\det(D_\lambda)(1 - w_N^N v^\top D^{-1}_\lambda v) = 0
$$
which simplifies to \eqref{eq:essential_condition}.
\item $D_\lambda$ is singular, and $\lambda = w_I^N$ for exactly one $N\neq I \subseteq N$. Then $\adjugate(D_\lambda) = \alpha e_Ie_I^\top$ for some nonzero $\alpha$ \cite[p.~22-23]{HornJohnson03}, where $(e_I)_J = 1$ if $I = J$ and $(e_I)_J = 0$ otherwise. Now \eqref{eq:adjoint_equality} simplifies to
$$
w_N^N v^\top(\alpha e_Ie_I^\top)v = 0
$$
which can only hold if $w_N^N = 0$. Hence such $\lambda$ cannot be a solution to \eqref{eq:characteristic_polynomial}.
\item $D_\lambda$ is singular and there are more than one $N\neq I \subseteq N$ such that $\lambda = w_I^N$. Then $\adjugate(D_\lambda) = 0$ \cite[p.~22]{HornJohnson03} and $\lambda$ is a solution to \eqref{eq:characteristic_polynomial}.
\end{enumerate}
\end{proof}

\begin{lemma} \label{lemma:necessary_and_sufficient_for_PSD}
Matrix $D + w_N^N vv^\top$ is positive-semidefinite if and only if either $w_I^N\geq 0$ for all $I\subseteq N$, or the following holds
\begin{align}
w_K^N < 0, & \qquad \textit{for exactly one } K \subseteq N,\\
w_J^N > 0, & \qquad \textit{for all } K \neq J \subseteq N, \\
\sum_{I \subseteq N} \frac{1}{w_I^N} \leq 0& \label{eq:sufficient_condition}
\end{align}
\end{lemma}
\begin{proof}
If $w_I^N\geq 0$ for all $I\subseteq N$ then $D + w_N^N vv^\top\succeq 0$ since it is the sum of two PSD matrices.
Otherwise, there exists $I\subseteq N$ with $w_I^N< 0$ and we distinguish between the following complementary cases.

If there are two different sets $N\neq I, J \subseteq N$ such that $w_I^N = w_J^N < 0$, then by Lemma \ref{lemma:PSD_condition} the matrix $D + w_N^N vv^\top$ has a negative eigenvalue. Therefore we may assume that all the negative entries of $D$ are different from each other. Then by Lemma~\ref{lemma:PSD_condition}, any potentially negative eigenvalue is given by \eqref{eq:essential_condition}. With this in mind, let $f(\lambda) = \sum_{N \neq I \subseteq N} \frac{1}{\lambda - w_I^N}$ and study the points $\lambda$ where $f(\lambda)$ intersects the line given by $\frac{1}{w^N_N}$.

 There are three cases:
 \begin{enumerate}
  \item For two sets $N \neq I,J \subseteq N$ we have $w_I^N < w_J^N \leq0$. Then since the function $f(\lambda)$ has vertical asymptotes at the points $w_I^N$ and $w_J^N$, there must be a point $\lambda < 0$ such that $f(\lambda) = \frac{1}{w_N^N}$ regardless of the value of $w^N_N$ (see Figure \ref{figure:conceptual_plot} (i)).

  \item For exactly one $N\neq I \subseteq N$ we have $w_I^N \leq 0$  and $w^N_N < 0$. Then $f(\lambda)$ has one vertical asymptote in $(-\infty, 0]$ and thus the line $\frac{1}{w^N_N}$ crosses the graph of $f(\lambda)$ at least in one $\lambda < 0$ (see Figure \ref{figure:conceptual_plot} (ii)).

  \item For exactly one $I \subseteq N$ we have $w_I^N <0 $ and the rest are strictly positive. Then we note that there can be at most one $\lambda < 0$ such that $f(\lambda) = \frac{1}{w_N^N}$. Inspecting the form of the graph shows that there is no intersection in the negative half-plane if and only if $f(0) \geq \frac{1}{w_N^N}$ (see Figure \ref{figure:conceptual_plot} (iii) and (iv) for the case $I=N$).
 \end{enumerate}
\end{proof}

\begin{figure}[h!]
\centering
\includegraphics[width=115mm]{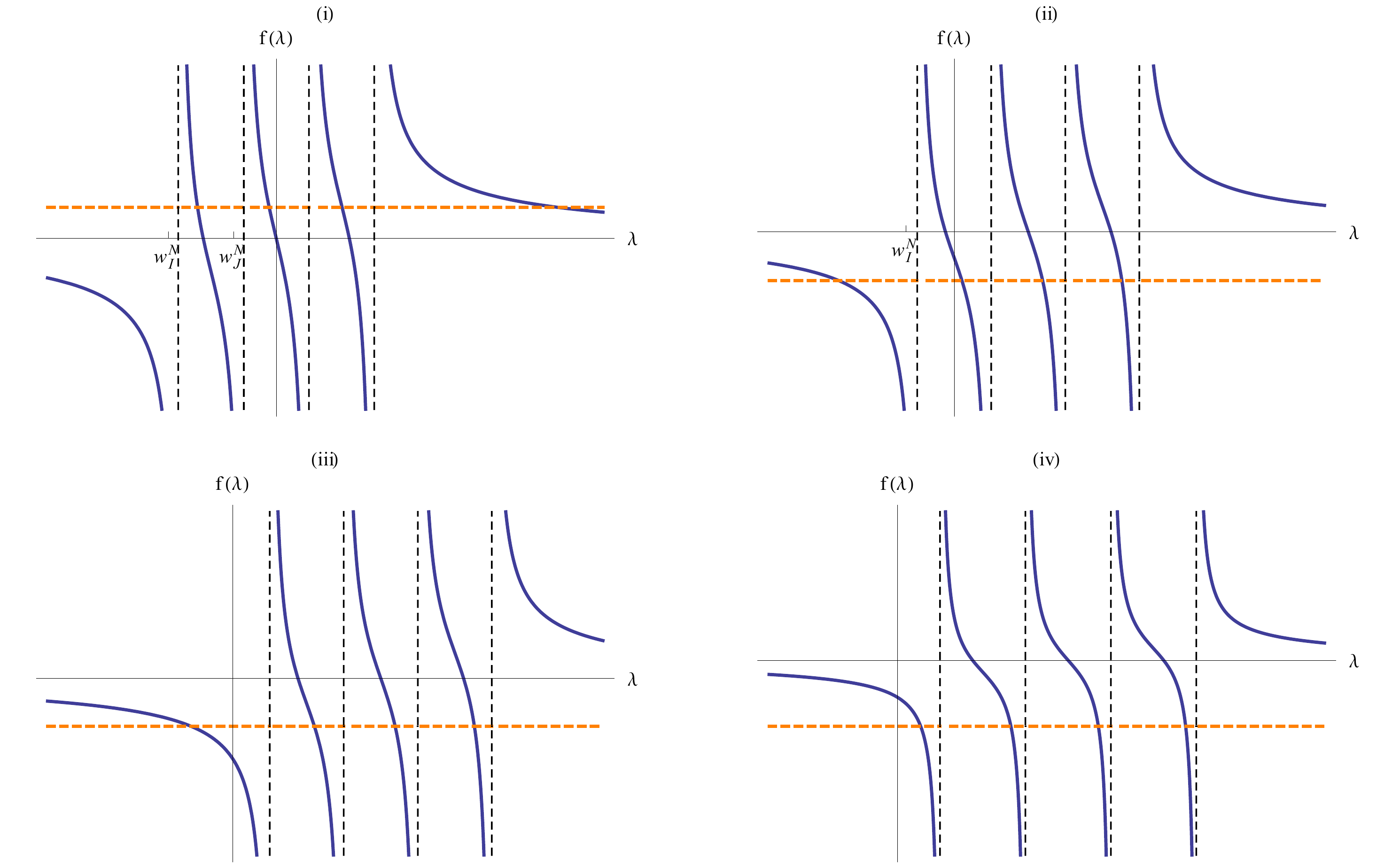}
\caption{A conceptual plot of different relevant arrangements of 
the graph of $f(\lambda)$ and the graph of $\frac{1}{w^N_N}$ (dotted lines).}
\label{figure:conceptual_plot}
\end{figure}

\section{Integrality gaps of Lasserre hierarchy at level $n-1$}

In this section we characterize the set of problems $\mathbb{P}$ of the form~\eqref{eq:polyproblem} that can have an integrality gap at level $n-1$ of the Lasserre relaxation.
In particular, we prove that in order to exhibit an integrality gap, a constrained problem can only have constraints each of which rule out only one point of the $\{0,1\}^n$ hypercube. We fully characterize what this means in the case where the constraints are linear. We also discuss two examples of problems with such constraints, and in particular, we exhibit a simple instance of the \textsc{Min Knapsack} problem that has an unbounded integrality gap. Finally, we show that if $\mathbb{P}$ is an unconstrained problem that has an integrality gap at level $n-1$, then the objective function of $\mathbb{P}$ must be a polynomial of degree $n$.
%

\subsection{Problems with linear constraints}

In this subsection we focus on $0/1$-integer linear programs $\mathbb{P}$ of the form \eqref{eq:polyproblem}. 
%
We will assume, w.l.o.g., that if constraint $g(x)\geq 0$ is satisfied by all integral points then it is redundant and no one of these redundant constraints is present. 


\begin{theorem} \label{th:gap_theorem}
Let $\mathbb{P}$ be a $0/1$-integer linear program of the form \eqref{eq:polyproblem}. 
The Lasserre relaxation $\LAS_{n-1}(\mathbb{P})$ has an integrality gap if and only if there exists a solution $\{y_I^N|I\subseteq N\}$ that satisfies the following conditions:
\begin{eqnarray} 
 y_I^N &>& 0 \qquad \text{for all } I\subseteq N,\label{posval}\\
\sum_{I \subseteq N} y_I^N &=& 1 \label{eq:gap_thm_sum_condition1}\\
g_\ell(x_{K_{\ell}})y_{K_{\ell}}^N &<& 0  \qquad  \text{ for exactly one } K_{\ell} \subseteq N \text{ for each } \ell\in [m], \label{eq:gap_thm_violated_constraint} \\
  g_\ell(x_J)y_J^N &>& 0 \qquad \text{for all } \ell\in [m], \text{ for all } K_{\ell} \neq J \subseteq N, \label{eq:gap_thm_not_violated_constraint}\\
  \sum_{I \subseteq N} \frac{1}{g_\ell(x_I)y_I^N} &\leq& 0 \qquad \text{for all } \ell\in [m], \label{eq:gap_thm_constrained_psd}\\
  \sum_{I \subseteq N} y_I^N f(x_I) &<& f(x_{I^*})   \label{eq:gap_theorem_gap_condition}
\end{eqnarray}
where $f(x_{I^*})$ is a minimal integral feasible solution.
\end{theorem}

\begin{proof}
We start justifying the feasibility conditions \eqref{posval}-\eqref{eq:gap_thm_constrained_psd}.
Consider a constrained problem with $m \geq 1$ constraints. By Lemma~\ref{prob}, at level $n-1$, the variables $y_I^N$ of a feasible solution must all be nonnegative and satisfy \eqref{eq:gap_thm_sum_condition1}, and hence can be seen as coefficients of convex combinations of the points $\set{0,1}^n$ (where $y_I^N$ is the coefficient of the solution $x_I$). 

If there is an integrality gap, then the projection of $y$ to the $\{0,1\}^n$ hypercube is not a convex combination of the feasible solutions of $\mathbb{P}$, and at least one variable $y_I^N$ must be positive such that the corresponding integer point $x_I$ is infeasible in the integer program. This implies that for some $\ell \in [m]$ and $J \subseteq N$ we must have $g_\ell(x_J)y_J^N < 0$. By Lemma \ref{lemma:necessary_and_sufficient_for_PSD}, \eqref{eq:gap_thm_violated_constraint}-\eqref{eq:gap_thm_constrained_psd} are necessary and sufficient for any such constraint in order to have the corresponding moment matrix PSD (this in particular implies that all $y_I^N > 0$ for all $I\subseteq N$ and~\eqref{posval} holds).

%
Next, assume that there is a constraint $\ell$ such $g_\ell(x_I)y_I^N\geq 0$ for each $I \subseteq N$. But then for every $I$, $g_{\ell}(x_{I}) \geq0$ (since $y_I^N>0$), which means that the constraint does not rule out any integer points of the set $\set{0,1}^n$ and therefore it is redundant to the problem.
Finally, Equation \eqref{eq:gap_theorem_gap_condition} is implied by the definition of the integrality gap. 
\end{proof}

\begin{definition}
 We call $g(x) \geq 0$ a \emph{Single Vertex Cutting} (SVC) constraint if there exists only one $I \subseteq N$ such that $g(x_I) < 0$ and for every other $I \neq J \subseteq N$ it holds $g(x_J) > 0$.
\end{definition}

\begin{corollary} \label{lemma:gap_problems_have_svc_constraints}
 Let $f(x_{I^*})$ denote the integral optimum of~\eqref{eq:polyproblem}. If there is an integrality gap, i.e., $y \in \LAS_{n-1}(\mathbb{P})$ such that $\sum_{I \subseteq N} y_I^N f(x_I) < f(x_{I^*})$, then the constraints in~\eqref{eq:polyproblem} are SVC.
\end{corollary}
\begin{proof}
 Assume there exists a solution $y \in \LAS_{n-1}(\mathbb{P})$ such that $\sum_{I \subseteq N} y_I^N f(x_I) < f(x_{I^*})$. Then by Theorem \ref{th:gap_theorem}, for any constraint $\ell$ equations \eqref{eq:gap_thm_violated_constraint}-\eqref{eq:gap_thm_constrained_psd} must hold. In particular, \eqref{eq:gap_thm_violated_constraint} and \eqref{eq:gap_thm_not_violated_constraint} imply that there can be only one violating assignment of the constraint $g_\ell(x)$, and no assignment can be such that $g_\ell(x) = 0$.
\end{proof}

We are considering only problems with linear constraints over $\set{0,1}^n$, so it is straightforward to characterize the SVC constraints.
\begin{lemma} \label{lemma:linear_svc_constraints}
 Let $g(x) = \sum_{i=1}^n a_ix_i - b \geq 0$ be a linear SVC constraint. Then $b \neq 0$ and $a_i \neq 0$ for all $i$, and if $P$ is the set of indices such that $a_i < 0 \Leftrightarrow i \in P$, then $\sum_{i \in P} a_i < b$, but $\sum_{i \in Q} a_i > b$ for all $P \neq Q \subseteq N$.
\end{lemma}
\begin{proof}
Let $g(x)$ be an SVC constraint. If $b = 0$, then $g(x_\varnothing) = 0$, which is not allowed by definition. Next we show that $a_i \neq 0$ for every $i = 1,...,n$: Assume this is not the case, and let $a_j = 0$. By definition, there is an $I \subseteq N$ such that $g(x_I) < 0$. But since $a_j = 0$, the variable $x_j$ is not present in the constraint $g(x)$, and thus also $g(x_{I \Delta \set{j}}) < 0$ (where $\Delta$ denotes the xor operator). Therefore $g(x)$ cannot be SVC. The remaining part follows from the definition of SVC constraint.
\end{proof}

\subsection{Example problems with SVC constraints at level $n-1$}

As proved in Corollary \ref{lemma:gap_problems_have_svc_constraints}, SVC constraints are in some sense the most difficult constraints to handle for the Lasserre hierarchy. Each such constraint excludes only one point of the $\{0,1\}^n$ hypercube, and thus the feasible set of any integer linear program can be modeled using only these constraints. It follows that if modeled in this way, any integer linear program can potentially have an integrality gap at level $n-1$ of the Lasserre hierarchy. In this section we give two examples of problems where the Lasserre hierarchy does not converge to the integer polytope even at level $n-1$.

\subsubsection{Unbounded integrality gap for the \textsc{Min Knapsack}}

One problem where the SVC constraint naturally arises is the \textsc{Knapsack} problem. We show that the minimization version of the problem has an unbounded integrality gap at level $n-1$ of the Lasserre hierarchy. Indeed, consider the following simple instance of the \textsc{Min-Knapsack}:
\begin{eqnarray}\label{LP:minKnapGap}
\begin{array}{rll}
(GapKnap)
 \min \{ \sum_{i=1}^n x_i| \sum_{i=1}^{n} x_i \geq 1/P, x_i \in \{0,1\} \mbox{ for } i\in[n]\}
\end{array}
\end{eqnarray}
Notice that the optimal \emph{integral} value of ($GapKnap$) is one. The optimal value of the linear programming relaxation of ($GapKnap$) is $1/P$, so the integrality gap of the LP is $P$ and can be arbitrarily large.

By using Theorem~\ref{th:gap_theorem} we prove the following dichotomy-type result. If we allow a ``large'' $P$ (exponential in the number of variables $n$), then the Lasserre hierarchy is of no help to limit the unbounded integrality  gap of ($GapKnap$), even at level $(n-1)$. This analysis is tight since $\LAS_{n}(GapKnap)$ admits an optimal integral solution with $n$ variables. We also show that the requirement that $P$ is exponential in $n$ is necessary for having a ``large'' gap at level $(n-1)$.

\begin{corollary}\label{th:knapgapbounds}
(Integrality Gap Bounds for \textsc{Min-Knapsack})
The integrality gap of $\LAS_{n-1}(GapKnap)$ is $k$, for any $k\geq 2$ if and only if $P= \Theta(k) \cdot 2^{2n}$.
%
\end{corollary}
\begin{proof}
We start by proving the 'if' direction. Consider the following solution with $P=k\cdot 2^{2n+1}$
\begin{eqnarray}
y^N_{I}&=& \frac{2^{n}}{P|I|-1} \qquad \forall I \subseteq N \text{ and } I\not = \emptyset \label{knapcond2}\\
y_{\emptyset}^N &=& 1-\sum_{I \subseteq N, I\not = \emptyset} y^N_{I} \label{knapcond1}
\end{eqnarray}
The value of the solution is equal to:
\begin{eqnarray*}
\sum_{i=1}^n y_i &=& \sum_{I\subseteq N} y_I^N |I| = \sum_{I\subseteq N}\frac{2^{n}}{P|I|-1} |I| \leq \frac{2^{2n+1}}{P}=\frac{1}{k}
\end{eqnarray*}
so the integrality gap is at least $k$. A direct computation shows that
\begin{eqnarray*}
\sum_{I \subseteq N} \frac{1}{y_{I}^Ng(x_I)} &=& -\frac{P}{y^N_{\emptyset}}
+\sum_{\emptyset \not = I \subseteq N} \frac{P}{y_{I}^N(P|I|-1)}= -\frac{P}{1-\sum_{I \subseteq N, I\not = \emptyset} y^N_{I}}
+\frac{P}{2^n}(2^n-1) < 0
\end{eqnarray*}
and hence, by Theorem~\ref{th:gap_theorem}, \eqref{knapcond2}-\eqref{knapcond1} is a feasible solution to $\LAS_{n-1}(GapKnap)$.

Next we prove the ``only if'' direction. Consider any feasible solution that creates an integrality gap: $\sum_{I\subseteq N} y_I^N |I|=1/k$ (for some $k\geq 2$). Then by Theorem~\ref{th:gap_theorem} we have that
\begin{eqnarray*}
\sum_{I \subseteq N} \frac{1}{y_{I}^Ng(x_I)} &=& -\frac{P}{y^N_{\emptyset}}
+\sum_{\emptyset \not = I \subseteq N} \frac{1}{y_{I}^N(|I|-1/P)}\leq 0
\end{eqnarray*}
For $k \geq 2$ the latter implies that
\begin{eqnarray*}
P&\geq& \left(1-\sum_{I \subseteq N, I\not = \emptyset} y^N_{I}\right)\sum_{\emptyset \not = I \subseteq N} \frac{1}{y_{I}^N(|I|-1/P)} \\
&>&\left(1-\frac{1}{k}\right)\sum_{\emptyset \not = I \subseteq N} \frac{1}{y_{I}^N|I|}= (k-1) (2^n-1)^2=\Theta(k) 2^{2n}
\end{eqnarray*}
The last inequality follows by observing that the minimum value of $\sum_{\emptyset \not = I \subseteq N} \frac{1}{y_{I}^N|I|}$ is when every summands is the same\footnote{The solution $x_i = \frac{1}{n}$, for each $i\in[n]$, is optimal to the optimization problem
$
\min \set{\sum_{i=1}^n \frac{1}{x_i} | \sum_{i=1}^n x_i = 1, x_i \geq 0}
$.\label{footnote}
}, namely $y_I^N |I|= \frac{1}{k(2^n-1)}$, since $\sum_{I\subseteq N} y_I^N |I|=1/k$ and $y_I^N>0$.

\end{proof}
\begin{remark}
We observe that the instance~\eqref{LP:minKnapGap} can be easily ruled out by requiring that each coefficient of any variable must be not larger than the constant term in the knapsack constraint.
However, even with this pruning step, the integrality gap can be made unbounded up to the last but two levels of the Lasserre hierarchy: add an additional variable $x_{n+1}$ only in the constraint (not in the objective function) and increase the constant term to $1+1/P$. 
Any solution for $\LAS_{n-1}(GapKnap)$ can be easily turned into a feasible solution for the augmented instance by setting the new variables $y_I'=y_{I\setminus\{n+1\}}$ for any $I\in \PS_{2t+2}([n+1])$ and observing that any principal submatrix of the new moment matrices has either determinant equal to zero or it is a principal submatrix in the moment matrix of the reduced problem.
\end{remark}

\subsubsection{Undetected empty integer hull}

As discussed at the beginning of this section, any integer linear problem can be modeled using SVC constraints. Formulating the problem in this ``pathological'' way can potentially hinder the convergence of the Lasserre hierarchy. We demonstrate this by showing an extreme example, where the Lasserre hierarchy cannot detect that the integer hull is empty even at level $n-1$.

Consider the feasible set given by (exponentially many) inequalities of the form
\begin{equation} \label{eq:empty_example_constraint}
\sum_{i \in P} (1-x_i)  + \sum_{i \in N\setminus P} x_i \geq b
\end{equation}
for each $P \subseteq N$. Clearly, any integral assignment $I$ such that $x_i = 1$ if $i \in I$ and $x_i = 0$ otherwise, cannot satisfy all of the inequalities when $b$ is positive. However, there exists an assignment of the variables $y_I^N$ that satisfies the conditions of Theorem \ref{th:gap_theorem}, and is hence a feasible solution to the Lasserre relaxation of the polytope described above at level $n-1$, as shown below.

Consider a symmetric solution $y_I^N = \frac{1}{2^n}$ for every $I \subseteq N$ and some constraint of the form \eqref{eq:empty_example_constraint} corresponding to a given set $P \subseteq N$. Now the variables $z_I^N = g(x_I)y_I^N$ satisfy \eqref{eq:gap_thm_violated_constraint} and \eqref{eq:gap_thm_not_violated_constraint}, and we need to check that it is possible to satisfy \eqref{eq:gap_thm_constrained_psd}:
\begin{eqnarray*}
\sum_{I \subseteq N} \frac{1}{z_I^N} = \frac{1}{2^n}\sum_{I \subseteq N} \frac{1}{|P \setminus I|+ |I\setminus P| - b} \leq 0
\Leftrightarrow \sum_{\varnothing \neq I \subseteq N} \frac{1}{|I| - b} \leq \frac{1}{b}
\end{eqnarray*}
When $0< b < \frac{1}{2}$, the above is implied by
$
\sum_{\varnothing \neq I \subseteq N} 2 \leq \frac{1}{b}
$,
so choosing $b = \frac{1}{2^{n+1}}$ makes \eqref{eq:gap_thm_constrained_psd} satisfied.

\subsection{Unconstrained problems at level $n-1$}
Let $f:\{0,1\}^n \rightarrow \mathbb{R}$ be an objective function of a polynomial minimization problem normalized such that $ \min_{x \in \{0,1\}^n} f(x) = 0$ and $ \max_{x \in \{0,1\}^n} f(x) = 1$. We start with the conditions that an unconstrained polynomial optimization problem has to satisfy in order do admit a gap at level $n-1$.

%
\begin{theorem} \label{th:gap_theorem_unconstrained}
Let $\mathbb{P}$ denote an unconstrained polynomial optimization problem of the form \eqref{eq:polyproblem}. 
The Lasserre relaxation $\LAS_{n-1}(\mathbb{P})$ has an integrality gap if and only if there exists a solution $\{y_I^N|I\subseteq N\}$ that satisfies~\eqref{eq:gap_theorem_gap_condition} and the following conditions:
\begin{eqnarray}
\sum_{I \subseteq N} y_I^N &=& 1 \label{eq:gap_thm_sum_condition}\\
y_K^N &<& 0 \qquad \textit{for exactly one } K \subseteq N, \label{eq:gap_thm_negative_variable}\\
y_J^N &>& 0 \qquad \textit{for all } K \neq J \subseteq N, \label{eq:gap_thm_positive_variables}\\
\sum_{I \subseteq N} \frac{1}{y_I^N} &\leq& 0 \label{eq:gap_thm_unconstrained_psd}
\end{eqnarray}
\end{theorem}
\begin{proof}
The proof is the unconstrained analogue of the proof of Theorem~\ref{th:gap_theorem}.
\end{proof}

We note that $f$ can always be represented as a multivariate polynomial of degree at most $n$. The main result of this section is Theorem~\ref{th:polygap}. 

\begin{theorem}\label{th:polygap}
 If $f$ is a function such that $f$ has an integrality gap at level $n-1$, then $f$ is a multivariate polynomial of degree $n$.
\end{theorem}
\begin{proof}
We will use some elementary Fourier analysis of boolean functions (see e.g.~\cite[Ch. 1]{ODonnell14}). To follow an established convention, we switch from studying the function $f:\set{0,1}^n \rightarrow \mathbb{R}$ to $h:\set{-1,1}^n \rightarrow \mathbb{R}$ via the bijective transform $f(x) = h(1-2x)$. Observe that $f$ is of degree $t$ if and only if $h$ is of degree $t$, and for any $S \subseteq N$ we have $f(x_S) = h(w_S)$, where $w_i = -1$ if $i \in S$ and $w_i = 1$ otherwise.

Assume as before that for some $I_1 \subseteq N$, $h(w_{I_1}) = 1$ and $0 \leq h(w_I) \leq 1$. We assume that $|I_1|$ is even and let $I_2 \subseteq N$ be some fixed set such that $|I_2|$ is odd (the case where $|I_1|$ is odd is symmetric). We assume that $h$ has an integrality gap, so by Lemma~\ref{lemma:no_gap_condition} (see below) necessarily $\sum_{I \subseteq N} h(w_I) < 2$, which we rewrite in a more convenient form (using $h(w_{I_1}) = 1$)
\begin{equation} \label{eq:assumption_on_h}
h(w_{I_2}) < 1 - \sum_{I_1 \neq I \neq I_2} h(w_I)
\end{equation}

Assume now that $h$ has a degree smaller than $n$, or in other words, its Fourier coefficient $\hat{h}(N)$ is 0:
$$
\hat{h}(N) = 2^{-n} \sum_{S \subseteq N} h(w_S) (-1)^{|S|} = 0
$$
Removing the normalizing constant and reordering the sum the above implies (using the assumptions on the parity of $|I_1|, |I_2|$)
$$
 \sum_{\substack{S \neq I_1 \\ |S| \text{ even}}} h(w_S) - \sum_{\substack{S \neq I_2 \\ |S| \text{ odd}}} h(w_S) = -1 + h(w_{I_2}) < - \sum_{I_1 \neq I \neq I_2} h(w_I)
$$
by \eqref{eq:assumption_on_h}. Moving all the $h$ terms to the left hand side yields
$$
2\sum_{\substack{S \neq I_1 \\ |S| \text{ even}}} h(w_S) < 0
$$
which contradicts the assumption that $h(w) \geq 0$.
\end{proof}
\begin{lemma} \label{lemma:no_gap_condition}
Let $f(x_{I_1}) = 1$ and $0 \leq f(x) \leq 1$ for every $x \in \set{0,1}^n$. If $f$ is such that
$
\sum_{I \subseteq N} f(x_I) \geq 2
$
then there is no gap at level $n-1$.
\end{lemma}
\begin{proof}
The condition on the values of $f$ means loosely speaking that we can ``reassign'' the values of $f$ such that we obtain a new function which satisfies the conditions in Lemma \ref{lemma:inclusion_exclusion_polynomial} (see below). More precisely, if $f$ has an integrality gap, we show that there exists another function satisfying the conditions of Lemma \ref{lemma:inclusion_exclusion_polynomial} that must also have an integrality gap.

Formally, assume $f(x)$ has a gap given by the variables $y_I^N$. We may assume that $y_{I_1}^N < 0$, and let $I_2 \neq I_1$ denote the set such that $y_{I_2}^N$ is the smallest positive variable of the variables $y_I^N$. Next, we define a function $\bar{f}$ as follows: $\bar{f}(x_{I_1})=\bar{f}(x_{I_2}) = 1$, $\bar{f}(x_I) = 0$ for every other $I_1 \neq I \neq I_2$. We show that now $\bar{f}$ must have a gap with the variables $\bar{y}_I^N$ given by
$$
\begin{aligned}
  \bar{y}_{I_1}^N = -\bar{\epsilon},&\\
  \bar{y}_{I_2}^N =   y_{I_2}^N,&\\
  \bar{y}_{I}^N = \frac{1+\bar{\epsilon} - \bar{y}_{I_2}^N}{2^n-2},& \text{ for all the other }I
\end{aligned}
$$
Here we set $\bar{\epsilon}$ such that the inequality \eqref{eq:gap_thm_unconstrained_psd} is tight, so that everything is determined and the equations \eqref{eq:gap_thm_negative_variable}-\eqref{eq:gap_thm_unconstrained_psd} are satisfied. By symmetry,
we have that $\bar{y}_{I_1}^N \leq y_{I_1}^N$.
Now we get that the objective value of the function $\bar{f}$ with variables $\bar{y}_I^N$ at level $n-1$ is
$$
\begin{aligned}
 \sum_{I \subseteq N} \bar{y}_I^N \bar{f}(x_I) = \bar{y}_{I_1}^N + \bar{y}_{I_2}^N \leq y_{I_1}^N + y_{I_2}^N \leq y_{I_1}^N f(x_{I_1}) + y_{I_2}^N \sum_{I_1 \neq I \subseteq N} f(x_I) \\
 \leq y_{I_1}^N f(x_{I_1}) + \sum_{I_1 \neq I \subseteq N} y_I^N f(x_I) = \sum_{I \subseteq N} y_I^N f(x_I) < 0
\end{aligned}
$$
where the second inequality follows directly from assuming that $f(x_{I_1}) = 1$ and that $\sum_{I \subseteq N} f(x_I) \geq 2$, and the third inequality from the assumption that $y_{I_2}$ is the smallest positive variable $y_I^N$. This shows that $\bar{f}$ has a gap with variables $\bar{y}_I^N$, which contradicts Lemma~\ref{lemma:inclusion_exclusion_polynomial}.
\end{proof}

The following lemma is used in the proof of Lemma~\ref{lemma:no_gap_condition}.
\begin{lemma} \label{lemma:inclusion_exclusion_polynomial}
 The function $f(x):\{0, 1\}^n \rightarrow \{0, 1\}$ such that for some $I_1 \neq I_2 \subseteq N$ we have $f(x_{I_1}) = f(x_{I_2}) = 1$ and $f(x_I) = 0$ for all $I_1 \neq I \neq I_2$ has no gap at level $n-1$.
\end{lemma}
\begin{proof}
Assume that the variables $\{y_I^N\}$ are such that there is an integrality gap. By Theorem~\ref{th:gap_theorem_unconstrained} only one variable gets a negative value whereas the others are positive. Moreover, the presence of a gap implies that $y_{I_1}^N<0$ (or equivalently $y_{I_2}^N<0$; Note that any other possibility does not give any integrality gap). Let $y_{I_1}^N = -\epsilon$ for some $\epsilon > 0$. Furthermore, we set $y_{I_2}^N = \delta$ for some $\delta > 0$ and due to symmetry, we set $y_I^N = \frac{1+\epsilon - \delta}{2^n-2}$ for all ${I_1} \neq I \neq {I_2}$. We note that, using any other asymmetric assignment would make the left hand side of the sum \eqref{eq:gap_thm_unconstrained_psd} larger (see Footnote~\ref{footnote}), which means we would have to choose a smaller $\epsilon$ which would yield a smaller gap.

The condition that we have an integrality gap ($\sum_{I \subseteq N} y_I^N f(x_I) <0$) requires that
$
\epsilon > \delta
$.
By condition~\eqref{eq:gap_thm_unconstrained_psd} we have:
$$
\sum_{I \subseteq N} \frac{1}{y_I^N} = \frac{1}{-\epsilon} + (2^n-2)\frac{1}{\frac{1+\epsilon-\delta}{2^n-2}} + \frac{1}{\delta} = \frac{1}{-\epsilon} + (2^n-2)^2\frac{1}{1+\epsilon-\delta} + \frac{1}{\delta} \leq 0
$$
This simplifies to
$$
(2^n-2)^2\frac{1}{1+\epsilon-\delta}  \leq \frac{1}{\epsilon} - \frac{1}{\delta} = \frac{\delta - \epsilon}{\delta \epsilon}
$$
If here $\epsilon > \delta$, the right hand side is always negative, and the left hand side positive. This is a contradiction and hence $f(x)$ cannot have an integrality gap at level $n-1$.
\end{proof}
We point out that there exists a function of degree $n$ that exhibits an integrality gap at level $n-1$. Consider the function given by
$$
f(x) = 1 - \sum_{\varnothing \neq I \subseteq N} (-1)^{|I|}\prod_{i \in I}x_i
$$
This function has the value 1 when all the variables are 0, and 0 elsewhere. It is a straightforward application of Theorem~\ref{th:gap_theorem_unconstrained} to show that $f(x)$ exhibits an integrality gap at level $n-1$. We remark that $f(x)$ can be seen as a constraint satisfaction version of an SVC constraint.





\paragraph{Acknowledgments.} Research supported by the Swiss National Science Foundation project ${200020\_ 144491\slash 1}$ ``Approximation Algorithms for Machine Scheduling Through Theory and Experiments''.
We thank Ola Svensson and the anonymous reviewers  for helpful comments.

{\small
\bibliographystyle{abbrv}
\bibliography{kih}

\begin{thebibliography}{10}

\bibitem{AroraBS10}
S.~Arora, B.~Barak, and D.~Steurer.
\newblock Subexponential algorithms for unique games and related problems.
\newblock In {\em FOCS}, pages 563--572, 2010.

\bibitem{AroraRV09}
S.~Arora, S.~Rao, and U.~V. Vazirani.
\newblock Expander flows, geometric embeddings and graph partitioning.
\newblock {\em Journal of the ACM}, 56(2), 2009.

\bibitem{BarakBHKSZ12}
B.~Barak, F.~G. S.~L. Brand{\~a}o, A.~W. Harrow, J.~A. Kelner, D.~Steurer, and
  Y.~Zhou.
\newblock Hypercontractivity, sum-of-squares proofs, and their applications.
\newblock In {\em STOC}, pages 307--326, 2012.

\bibitem{BarakCK15}
B.~Barak, S.~O. Chan, and P.~Kothari.
\newblock Sum of squares lower bounds from pairwise independence.
\newblock In {\em STOC}, 2015.

\bibitem{BarakRS11}
B.~Barak, P.~Raghavendra, and D.~Steurer.
\newblock Rounding semidefinite programming hierarchies via global correlation.
\newblock In {\em FOCS}, pages 472--481, 2011.

\bibitem{BateniCG09}
M.~Bateni, M.~Charikar, and V.~Guruswami.
\newblock Maxmin allocation via degree lower-bounded arborescences.
\newblock In {\em STOC}, pages 543--552, 2009.

\bibitem{BhaskaraCVGZ12}
A.~Bhaskara, M.~Charikar, A.~Vijayaraghavan, V.~Guruswami, and Y.~Zhou.
\newblock Polynomial integrality gaps for strong sdp relaxations of densest
  {\it k}-subgraph.
\newblock In {\em SODA}, pages 388--405, 2012.

\bibitem{Cheung07}
K.~K.~H. Cheung.
\newblock Computation of the {L}asserre ranks of some polytopes.
\newblock {\em Mathematics of Operations Research}, 32(1):88--94, 2007.

\bibitem{Chlamtac07}
E.~Chlamtac.
\newblock Approximation algorithms using hierarchies of semidefinite
  programming relaxations.
\newblock In {\em FOCS}, pages 691--701, 2007.

\bibitem{ChlamtacS08}
E.~Chlamtac and G.~Singh.
\newblock Improved approximation guarantees through higher levels of {SDP}
  hierarchies.
\newblock In {\em APPROX-RANDOM}, pages 49--62, 2008.

\bibitem{Chla12}
E.~Chlamtac and M.~Tulsiani.
\newblock Convex relaxations and integrality gaps.
\newblock In {\em to appear in Handbook on semidefinite, conic and polynomial
  optimization}. Springer.

\bibitem{DBLP:conf/soda/CyganGM13}
M.~Cygan, F.~Grandoni, and M.~Mastrolilli.
\newblock How to sell hyperedges: The hypermatching assignment problem.
\newblock In {\em SODA}, pages 342--351, 2013.

\bibitem{VegaK07}
W.~F. de~la Vega and C.~Kenyon-Mathieu.
\newblock Linear programming relaxations of maxcut.
\newblock In {\em SODA}, pages 53--61, 2007.

\bibitem{FawziSaundersonParrilo15}
H.~Fawzi, J.~Saunderson, and P.~Parrilo.
\newblock Sparse sum-of-squares certificates on finite abelian groups.
\newblock {\em CoRR}, abs/1503.01207, 2015.

\bibitem{GoemansW95}
M.~X. Goemans and D.~P. Williamson.
\newblock Improved approximation algorithms for maximum cut and satisfiability
  problems using semidefinite programming.
\newblock {\em Journal of the ACM}, 42(6):1115--1145, 1995.

\bibitem{GouveiaPT13}
J.~Gouveia, P.~A. Parrilo, and R.~R. Thomas.
\newblock Theta bodies for polynomial ideals.
\newblock {\em {SIAM} Journal on Optimization}, 20(4):2097--2118, 2010.

\bibitem{Grigoriev01}
D.~Grigoriev.
\newblock Complexity of positivstellensatz proofs for the knapsack.
\newblock {\em Computational Complexity}, 10(2):139--154, 2001.

\bibitem{Grigoriev01b}
D.~Grigoriev.
\newblock Linear lower bound on degrees of positivstellensatz calculus proofs
  for the parity.
\newblock {\em Theoretical Computer Science}, 259(1-2):613--622, 2001.

\bibitem{GrigorievHP02}
D.~Grigoriev, E.~A. Hirsch, and D.~V. Pasechnik.
\newblock Complexity of semi-algebraic proofs.
\newblock In {\em STACS}, pages 419--430, 2002.

\bibitem{GrigorievV01}
D.~Grigoriev and N.~Vorobjov.
\newblock Complexity of null-and positivstellensatz proofs.
\newblock {\em Annals of Pure and Applied Logic}, 113(1-3):153--160, 2001.

\bibitem{GuruswamiS11}
V.~Guruswami and A.~K. Sinop.
\newblock {L}asserre hierarchy, higher eigenvalues, and approximation schemes
  for graph partitioning and quadratic integer programming with psd objectives.
\newblock In {\em FOCS}, pages 482--491, 2011.

\bibitem{HornJohnson03}
R.~A. Horn and C.~R. Johnson.
\newblock {\em Matrix analysis}.
\newblock Cambridge University Press, 2013.

\bibitem{KarlinMN11}
A.~R. Karlin, C.~Mathieu, and C.~T. Nguyen.
\newblock Integrality gaps of linear and semi-definite programming relaxations
  for knapsack.
\newblock In {\em IPCO}, pages 301--314, 2011.

\bibitem{Lasserre01}
J.~B. Lasserre.
\newblock Global optimization with polynomials and the problem of moments.
\newblock {\em SIAM Journal on Optimization}, 11(3):796--817, 2001.

\bibitem{Laurent03}
M.~Laurent.
\newblock A comparison of the {S}herali-{A}dams, {L}ov{\'a}sz-{S}chrijver, and
  {L}asserre relaxations for 0-1 programming.
\newblock {\em Mathematics of Operations Research}, 28(3):470--496, 2003.

\bibitem{laurent09}
M.~Laurent.
\newblock Sums of squares, moment matrices and optimization over polynomials.
\newblock {\em Emerging Applications of Algebraic Geometry}, (149):157--270,
  2009.

\bibitem{LeeRagSteu15}
J.~R. Lee, P.~Raghavendra, and D.~Steurer.
\newblock Lower bounds on the size of semidefinite programming relaxations.
\newblock {\em to appear in STOC 2015}.

\bibitem{Lovasz79}
L.~Lov\'asz.
\newblock On the shannon capacity of a graph.
\newblock {\em IEEE Transactions on Information Theory}, 25:1--7, 1979.

\bibitem{MagenM09}
A.~Magen and M.~Moharrami.
\newblock Robust algorithms for on minor-free graphs based on the
  {S}herali-{A}dams hierarchy.
\newblock In {\em APPROX-RANDOM}, pages 258--271, 2009.

\bibitem{Nesterov00}
Y.~Nesterov.
\newblock {\em Global quadratic optimization via conic relaxation}, pages
  363--384.
\newblock Kluwer Academic Publishers, 2000.

\bibitem{ODonnell14}
R.~O'Donnell.
\newblock {\em Analysis of Boolean Functions}.
\newblock Cambridge University Press, 2014.

\bibitem{ODonnellZ13}
R.~O'Donnell and Y.~Zhou.
\newblock Approximability and proof complexity.
\newblock In {\em SODA}, pages 1537--1556, 2013.

\bibitem{parrilo00}
P.~Parrilo.
\newblock {\em Structured Semidefinite Programs and Semialgebraic Geometry
  Methods in Robustness and Optimization}.
\newblock {PhD} thesis, California Institute of Technology, 2000.

\bibitem{RaghavendraT12}
P.~Raghavendra and N.~Tan.
\newblock Approximating csps with global cardinality constraints using sdp
  hierarchies.
\newblock In {\em SODA}, pages 373--387, 2012.

\bibitem{Rot13}
T.~Rothvo{\ss}.
\newblock The lasserre hierarchy in approximation algorithms.
\newblock Lecture Notes for the MAPSP 2013 - Tutorial, June 2013.

\bibitem{Schoenebeck08}
G.~Schoenebeck.
\newblock Linear level {L}asserre lower bounds for certain k-csps.
\newblock In {\em FOCS}, pages 593--602, 2008.

\bibitem{schor87}
N.~Shor.
\newblock Class of global minimum bounds of polynomial functions.
\newblock {\em Cybernetics}, 23(6):731--734, 1987.

\bibitem{Tulsiani09}
M.~Tulsiani.
\newblock Csp gaps and reductions in the {L}asserre hierarchy.
\newblock In {\em STOC}, pages 303--312, 2009.

\end{thebibliography}
}

\end{document}